\documentclass[cmp]{svjour}

\usepackage[utf8]{inputenc}
\usepackage[T1]{fontenc}
\usepackage[pdftex]{graphicx}
\usepackage{amsmath}
\usepackage{amsfonts}
\usepackage{amssymb}
\usepackage{dsfont}
\usepackage{stmaryrd}
\usepackage{subfigure}
\usepackage{url}
\usepackage{wasysym}
\usepackage{color}
\usepackage{stackrel}
\usepackage{bbm}
\newcommand{\La}{{\mathtt L}}
\newcommand{\Ra}{{\mathtt R}}
\newcommand{\mm}[1]{{\mathbf{#1}}}

\newcommand{\ii}{ {\rm i} }
\newcommand{\dd}{ {\rm d} }
\newcommand{\ZZ}{\mathbb{Z}}
\newcommand{\RR}{\mathbb{R}}
\newcommand{\CC}{\mathbb{C}} 
\newcommand{\AAA}{{\mathfrak A}}
\def\one{\mathbbm{1}}

\newcommand{\ave}[1]{{\left\langle #1\right\rangle}}

\usepackage[unicode]{hyperref} 

\newcommand{\tr}{\text{tr}}
\def\bra#1{\mathinner{\langle{#1}|}}
\def\ket#1{\mathinner{|{#1}\rangle}}

\def\x{{\rm x}}
\def\y{{\rm y}}
\def\z{{\rm z}}

\journalname{Communications in Mathematical Physics}

\begin{document}

\title{Thermodyamic bounds on Drude weights in terms of almost-conserved quantities}
\titlerunning{Thermodyamic bounds on Drude weights in terms of almost-conserved quantities}

\author{Enej Ilievski and Toma\v z Prosen}
\institute{Department of Physics, Faculty of Mathematics and Physics, University of Ljubljana\\
Jadranska 19, SI-1000 Ljubljana, Slovenia\\\\
\email{enej.ilievski@fmf.uni-lj.si, tomaz.prosen@fmf.uni-lj.si}
}
\authorrunning{E. Ilievski and T. Prosen}

\date{\today}

\maketitle

\begin{abstract}
We consider one-dimensional translationally invariant quantum spin (or fermionic) lattices and prove a Mazur-type inequality bounding the time-averaged thermodynamic limit of a finite-temperature expectation of a spatio-temporal autocorrelation function of a local observable in terms of quasi-local conservation laws with open boundary conditions. Namely, the commutator between the Hamiltonian and the conservation law of a finite chain may result in boundary terms only.
No reference to techniques used in Suzuki's proof of Mazur bound is made (which strictly applies only to finite-size systems with exact conservation laws), but Lieb-Robinson bounds and exponential clustering theorems of quasi-local $C^*$ quantum spin algebras are invoked instead.
Our result has an important application in the transport theory of quantum spin chains, in particular it provides rigorous non-trivial examples of
positive finite-temperature spin Drude weight in the anisotropic Heisenberg $XXZ$ spin 1/2 chain [Phys. Rev. Lett. {\bf 106}, 217206 (2011)].
\end{abstract}

\section{Introduction}

\subsection{The problem}

In 1969, Peter Mazur suggested \cite{Mazur} that {\em the time-average}
$\bar{A}=\lim_{t\to\infty}t^{-1} \int_0^t \dd t' A(t')$ of a bounded observable $A(t)$ can be bounded from below
by means of {\em exact conservation laws}, i.e. observables $Q_k$, $k=1,2\ldots$, 
satisfying $(\dd/\dd t)Q_k=0$ and being mutually in {\em involution}, namely
\begin{equation}
\ave{\bar{A}^2}_\beta = \lim_{t\to\infty}\frac{1}{t}\int_{0}^t \dd t' \ave{A(0)A(t')}_\beta \ge \sum_k \frac{\ave{A Q_k }^2_\beta}{\ave{ Q_k^2}_\beta}.
\label{eq:MazursIneq}
\end{equation}
$\ave{\bullet}_\beta$ is a thermal average at inverse temperature $\beta$, and observables $Q_k$ have to be chosen to be mutually
`orthogonal' $\ave{Q_k Q_l}_\beta = \delta_{k,l} \ave{Q_k^2}_\beta$.
Considering observables with a vanishing equilibrium expectation $\ave{A}_{\beta} = 0$, strict positivity of the right-hand side (RHS) of Mazur inequality (\ref{eq:MazursIneq}) is a convenient indicator of {\em non-ergodicity} of the observable $A$.
Mazur has shown that in classical statistical physics, the inequality (\ref{eq:MazursIneq}) is merely a corollary of the 
Khinchin theorem \cite{Khinchin}. Later, Suzuki \cite{Suzuki} has proven a quantum version of the bound (\ref{eq:MazursIneq}),
strictly applying only to {\em finite} quantum systems as his proof is based on explicit diagonalization of the Hamiltonian operator.
Existence of non-trivial constants of motion which is characteristic of {\em completely integrable systems} implies non-ergodicity of (almost all/generic) observables, making an intimate connection between non-ergodicity and complete integrability, both in classical and quantum statistical mechanics.

The inequality (\ref{eq:MazursIneq}) has found numerous and very useful applications in condensed matter physics as it is naturally 
suited for bounding dynamical susceptibilities within the linear response theory. For example, within the framework of Kubo's linear response approach, the zero-frequency Drude peak (see e.g.\cite{Meisner}) is defined in terms of the real part of (heat/electric/spin) conductivity $\sigma'_\beta(\omega)= 2\pi D_\beta \delta(\omega) + \sigma^{\rm reg}_\beta(\omega)$. The constant $D_\beta$ known as the {\em Drude weight} can be expressed\footnote{For a detailed discussion and derivation of the linear response expression of Drude weight see subsection \ref{sect:lr}.} for a one-dimensional quantum lattice of size $n$ as
\begin{equation}
D_\beta = \lim_{t\to\infty}\lim_{n\to\infty} \frac{\beta}{4 n t} \int_{-t}^t \dd t' \ave{J_n(0) J_n(t')}_\beta .
\label{eq:Kubo}
\end{equation}
The symmetrized correlation function is used in order to render Drude weight manifestly real.
Here, in (\ref{eq:Kubo}), $J_n=\sum_{x=1}^n j_x$ is an extensive (energy/particle/spin) current operator, and $j_x$ is a current density at site $x$.
Positivity of the Drude weight $D_\beta > 0$ is a signature of a ballistic transport at finite temperature and is generically related
to complete integrability via Mazur inequality (\ref{eq:MazursIneq}), as pointed out in Ref.~\cite{Zotos}.

It should be emphasized that the order of limits in the Kubo-type formula (\ref{eq:Kubo}) is crucial, namely taking them in a different order can sometimes produce completely different results (see e.g. Ref.~\cite{Rigol}).
Within Suzuki's approach, {\em thermodynamic} limit (TL) of letting the system size $n$, or the number of degrees of freedom, $n\to\infty$,
can only be taken at the end, which is not according to a fundamental principle of statistical mechanics which states that
TL $n\to\infty$ has to be always considered before the {\em long-time} limit $t \to \infty$ (say as in formula 
(\ref{eq:Kubo})).
Our program in this paper is then to use the natural language of quantum statistical mechanics of systems with local interactions --- the quasi-local $C^*$-algebras ---
and to develop an algebraic approach to Drude weight bounds of the Mazur-type pertaining to strictly infinite quantum lattice systems. In the Suzuki's theorem conservation laws have to strictly commute with the Hamiltonian, $[H,Q_k]=0$, for all finite sizes $n$. For that reason, one has to often study periodic boundary conditions which may only guarantee existence of such 
exact conserved quantities. However, in our $C^*$-algebraic setup we can accommodate also for quite common situations, where for any finite system size $n$ the conservation laws $Q_k$ might not be perfect, but the commutator $\dd Q_k/\dd t = \ii [H,Q_k]$ may result in terms supported at the system's boundaries. In this way, TL $n\to \infty$ can be taken in the beginning, by inclusion of larger and larger open lattices, and no resort to periodic boundary conditions is needed.

\subsection{Summary of the main results}

The main idea behind our susceptibility bounds is to bring together two classical results in quantum spin lattice systems with local interactions \cite{BR,Bruno,LiebRobinson,Araki}, namely (i) the exponential clustering property of finite temperature Gibbs states in one-dimensional translationally invariant lattices, and (ii)  Lieb-Robinson kinematic bound on the spatio-temporal propagation of local quantum correlations which result from boundary terms of the commutators. The projection form of the Lieb-Robinson bound \cite{BHV} enables us to nicely separate the causality light-cone of a time-evolved observable in the autocorrelator and to exploit the property of exponential clustering in the equilibrium state.

After introducing the notation and outlining the main concepts ((i) and (ii)) in section \ref{sect:prel} we write and prove in section \ref{sect:main} our main result (Theorem 1): Namely, given any observable $J_n$ and Hamiltonian $H_n$ of a finite lattice system of size $n$, such that the interactions are translationally invariant, and another quasi-local extensive observable $Q_n$ which has a property that the commutator $[H_n,Q_n] = b_1 - b_n$ where $b_1$ and $b_n$ are two operators supported at the left and the right boundary of the chain, we show that the Drude weight (\ref{eq:Kubo}) is strictly bounded by a simple expression 
$D_\beta \ge \frac{\beta}{2} \lim_{n\to \infty} n^{-1} (\frac{1}{2}\ave{J_n Q_n + Q_n J_n}_\beta)^2/\ave{Q_n^2}_\beta$.
In section \ref{sect:general} we then provide trivial generalization of the result (Theorem 2) to the case where we have an arbitrary set of almost-conserved quasi-local operators. In section \ref{sect:examples} we describe a nontrivial application of our results for bounding the spin Drude weight in the anisotropic Heisenberg $XXZ$ model. In section \ref{sect:discussion} we conclude by discussing the assumptions needed to equate the thermal-averaged correlator (\ref{eq:Kubo}) with the canonical Kubo-Mori expression, and state some remarks on possible other general contexts where results of our type may appear.

\section{Preliminaries}
\label{sect:prel}

We consider the following setup, where the notation essentially follows Bratteli and Robinson \cite{BR}. Let $N \in \ZZ_+$ be a local Hilbert space dimension, say $N=2$ for spins 1/2 or qubits,
and $\AAA_x \equiv \CC^{N \times N}$, $x \in \ZZ$, a local on-site matrix algebra.
We associate a matrix algebra to any finite open lattice of integers, the so-called {\em chain} $[x,y] = \{x,x+1,\ldots,y-1,y\}$, as
$\AAA_{[x,y]} = \bigotimes_{z=x}^y \AAA_z$ and define the quasi-local (UHF) $C^*$-algebra $\AAA=\AAA_\ZZ$ in terms of a closure of the {\em limit by inclusion} $[x,y] \to \ZZ$.

Let the {\em interaction} $h\in \AAA_{[0,d_{h}-1]}$ be an element of a local spin $C^{*}$-algebra on $d_h$ sites \footnote{We shall assume $d_h \ge 2$, as the on-site interaction with $d_h=1$ represent a trivial case with strictly local dynamics.}. The Hamiltonian
\begin{equation}
H_{\Lambda_{n}}=\sum_{x=1}^{n-d_{h}+1}h_{x}
\end{equation}
is an operator acting on a finite chain $\Lambda_{n} \equiv [1,n]$, which is a sum of local energy densities 
$h_{x}=\eta_{x}(h) \in \AAA_{[x,x+d_h-1]}$, obtained
by a group of lattice ({\em shift}) $*$-automorphisms $\eta_x$ of $\AAA$, defined by $\eta_{y}(a_{x})=a_{x+y}$.
Translationally invariant Hamiltonian can be understood in terms of the limit by inclusion $\Lambda\to\ZZ$ of Hamiltonians $H_\Lambda$ for arbitrary chains $\Lambda$
\begin{equation}
H_{\Lambda}=\sum_{x=\text{min}\,\Lambda}^{\text{max}\,\Lambda-d_{h}+1}h_{x}.
\end{equation}
The latter defines another group of $*$-automorphisms of the quasi-local algebra $\AAA$, namely the {\em time automorphism}
\begin{equation}
\tau_{t}(a)=\lim_{\Lambda \to \mathbb{Z}}\tau^{\Lambda}_{t}(a),\qquad \tau^{\Lambda}_{t}(a) := e^{\ii H_{\Lambda}t}ae^{-\ii H_{\Lambda}t},
\end{equation}
and a finite temperature equilibrium expectation, namely the infinite volume {\em Gibbs state} 
\begin{equation}
\omega_\beta(a) = \lim_{\Lambda\to\ZZ} \frac{\tr (a e^{-\beta H_{\Lambda}})}{\tr(e^{-\beta H_{\Lambda}})}
\end{equation}
strictly defined only for local operators $a$ and extended to $\AAA$ by continuity. Araki (Theorem 2.3 of \cite{Araki}) has shown that
such Gibbs state is an extremal $(\tau,\beta)$-KMS state,
which is invariant under space and time translations
\begin{eqnarray}
\omega_\beta(\eta_x(A)) &=& \omega_\beta(A), \label{eq:translx}\\
\omega_\beta(\tau_t(A)) &=& \omega_\beta(A),  \label{eq:translt}
\end{eqnarray}
for any $A\in\AAA$, $x\in\ZZ$, $t\in\RR$. Most importantly, Gibbs state $\omega_\beta$ has an {\em exponential clustering property} (ECP) 
(Theorem 2.3 of \cite{Araki}, see also Theorem 3 of \cite{Matsui}): for any pair of local operators $f \in \AAA_{[-d_f,-1]}$,  $g \in \AAA_{[0,d_g-1]}$,  $d_f,d_g\in \ZZ_+$, and a displacement $x\in \ZZ_+$
one has
\begin{equation}
\left|\omega_\beta(f \eta_{x}(g)) - \omega_\beta(f)\omega_\beta(g)\right| \le
\kappa \|f\| \|g\| e^{-\rho x}
\label{eq:ECP}
\end{equation}
where $\kappa,\rho$ are two positive constants, which do not depend on $x$, neither on $f,g$.

Next, we define an extensive current operator in the open chain
\begin{equation}
J_{\Lambda_{n}}=\sum_{x=1}^{n-d_{j}+1}j_{x},\quad j_{x}=\eta_{x}(j),
\label{eq:current}
\end{equation}
where the current density $j$ belongs to a $d_j$-site local algebra  $j\in \AAA_{[0,d_{j}-1]}$.
In fact, the observable $J$ may not necessarily be interpreted as a physical current, but it can be any spatial sum of a local self-adjoint operator 
$j$ (representing an extensive translationally invariant observable), the only condition being that its local equilibrium expectation vanishes
\begin{equation}
\omega_\beta(j) = 0.
\label{eq:jzero}
\end{equation} 

We shall think of it as a current merely because the most important application we have in mind is in the quantum transport.

\begin{definition}
\label{def:1}
The key concept in our work is a {\em quasi-local translationally invariant conservation law} $Q$ with the following properties:
\begin{enumerate}
\item $Q$ is a translationally invariant spatial sum of exponentially localized (quasi-local) operators.
For any finite chain $\Lambda_n$:
\begin{equation}
Q_{\Lambda_{n}}=\sum_{d=1}^{n}Q^{(d)}_{\Lambda_n},\quad Q_{\Lambda_{n}}^{(d)}=\sum_{x=1}^{n-d+1}q_{x}^{(d)},\quad \|q^{(d)}\|\leq \gamma \exp{(-\xi d)},
\label{eq:qdef}
\end{equation}
where $q^{(d)} = (q^{(d)})^* \in \AAA_{[0,d-1]}$, and $\gamma,\xi$ are positive $n$-independent constants.
\item
$Q$ should have vanishing thermal expectation value
\begin{equation}
\omega_\beta(Q_{\Lambda_n}) = 0.
\end{equation}
We can thus assume also that all orders of local density $q^{(d)}$ satisfy
\begin{equation}
\omega_\beta(q^{(d)})=0.
\end{equation}
\item
The operator $Q_{\Lambda_{n}}$ is \textit{almost-conserved} on any open chain $\Lambda_{n}$, i.e. it commutes with the Hamiltonian $H_{\Lambda_{n}}$ except for terms that are supported at the boundary of the chain
\begin{equation}
[H_{\Lambda_{n}},Q_{\Lambda_{n}}]=B_{\partial_{n}}
\label{eqn:commutator}
\end{equation}
where $\partial_n \equiv [1,d_b] \cup [n-d_b+1,n]$ and
\begin{equation}
B_{\partial_n} := b_{1}-b_{n-d_{b}+1}
\end{equation} 
for some local operator $b \in \AAA_{[0,d_b-1]}$.
\end{enumerate}
\end{definition}
For concrete, nontrivial examples of $Q$, see section \ref{sect:examples}.

Take now any chain $\Lambda$ which is sufficiently bigger than $\Lambda_n$, say $\Lambda \supseteq [-d_h + 2,n+d_h-1]$. Then, the almost-commutation identity (\ref{eqn:commutator})
implies
\begin{equation}
[H_\Lambda,Q_{\Lambda_n}] = B_{\partial_n} + [h_{\rm L},Q_{\Lambda_n}] + [h_{\rm R},Q_{\Lambda_n}]
\label{eq:comm2}
\end{equation}
where
\begin{equation}
h_{\rm L}:=\sum_{x=-d_{h}+2}^{0}h_{x},\qquad h_{\rm R}:=\sum_{x=n-d_{h}+2}^{n}h_{x}
\end{equation}
represent the left and the right near-boundary interactions.
The RHS of (\ref{eq:comm2}) can be rewritten as a sum of two quasi-local operators localized near the boundary
\begin{eqnarray}
[H_\Lambda,Q_{\Lambda_n}] &=& B_{\rm L} + B_{\rm R}, \label{eq:eom} \\
B_{\rm L} &:=& \sum_{d=0}^n b^{(d)}_{\rm L}, \nonumber\\ 
B_{\rm R} &:=& \sum_{d=0}^n b^{(d)}_{\rm R}, \\
b^{(0)}_{\rm L} &:=& b_1,\,\,\qquad\qquad  b^{(d)}_{\rm L} := \sum_{x=-d_h+2}^0\sum_{y=1}^{x+d_h-1} [h_x,q^{(d)}_y], \nonumber\\
b^{(0)}_{\rm R} &:=& -b_{n-d_b+1},\quad b^{(d)}_{\rm R} := \sum_{x=n-d_h+2}^n \sum_{y=n-d-d_h+3}^{x-d+1} [h_x,q^{(d)}_y]. 
\end{eqnarray}
Note that the supports\footnote{
Due to existence of the \textit{principle of locality} in the systems with finite-range interactions, the notion of an operator support enters naturally into
the discussion. The support ${\rm supp}\,A$ of an observale $A\in \AAA_\Lambda$ is the minimal set $\Gamma \subset \Lambda$ for which
$A=\tilde{A}\otimes \one_{\Lambda \setminus \Gamma}$ for some $\tilde{A}\in \AAA_\Gamma$.
} of the boundary operators are
\begin{eqnarray}
{\rm supp\,} b^{(d)}_{\rm L} &\subseteq& [-d_h+2,\max\{d_h-2+d,d_b\}], \nonumber \\ 
{\rm supp\,} b^{(d)}_{\rm R} &\subseteq& [\min\{n-d-d_h+3,n-d_b+1\},n+d_h-1],
\end{eqnarray}
and that they are {\em exponentially localized} (following from the definition (\ref{eq:qdef}) and elementary inequalities, $\| A B\| \le \| A\| \| B \|$ and the triangular inequality):
\begin{equation}
\| b^{(d)}_{\rm L} \| = \| b^{(d)}_{\rm R}\| \le \max\{ d_h(d_h-1) \gamma \| h\| ,\| b\|\}e^{-\xi d}.
\end{equation}
Eq. (\ref{eq:eom}) results in the Heisenberg equation of motion for the almost-conserved operator  
\begin{equation}
(\dd/\dd t)\tau^\Lambda_t(Q_{\Lambda_n}) = 
\ii [H_\Lambda, \tau^\Lambda_t(Q_{\Lambda_n})] = \ii \tau^\Lambda_t(B_{\rm L}+B_{\rm R}),
\label{eq:heom}
\end{equation} 
which together with the  initial condition $\tau_{t=0}(Q_{\Lambda_n}) = Q_{\Lambda_n}$, after taking the limit $\Lambda\to\ZZ$, integrates to an explicit time dependence
\begin{equation}
\tau_{t}(Q_{\Lambda_{n}})=Q_{\Lambda_{n}} + \ii \int_{0}^{t}\dd s\, \tau_{s}(B_{\rm L}+B_{\rm R}).
\label{eq:explicit}
\end{equation}

Another crucial technical tool that we shall facilitate is the {\em Lieb-Robinson estimate} (LRE) \cite{LiebRobinson} which bounds the speed at which a disturbance propagates through a quantum spin system with local interactions. Let $f \in \AAA_{X}$, 
$g \in \AAA_{\Gamma}$,  where $X,\Gamma\subset \ZZ$ are two subsets with $|X|$, $|\Gamma|$ sites, such that at least one of them is
finite. A useful form of LRE (see e.g. \cite{BHV}) then states
\begin{equation}
\|[\tau_{t}(f),g]\|\leq 
\phi \min\{|X|,|\Gamma|\} \| f\| \| g \| \exp{\left(-\mu({\rm dist}(X,\Gamma) - v |t|)\right)},
\label{eq:LRE}
\end{equation}
where ${\rm dist}(X,\Gamma)=\min_{x\in X, y\in\Gamma}|x-y|$ is the distance between the sets $X,\Gamma$,
and $\phi$, $\mu$ and $v$ are some positive constants, independent of $f,g$, and $t$. 

Regarding any finite subset $\Gamma \subset \ZZ$ we define a projection mapping $(\bullet)_\Gamma:\AAA \to \AAA$ as
\begin{equation}
(A)_\Gamma := \lim_{\Lambda\to\ZZ} \frac{ \tr_{\Lambda\setminus \Gamma}(A)}{\tr (\one_{\Lambda\setminus \Gamma})} \otimes \one_{\Lambda\setminus\Gamma}
= \lim_{\Lambda\to\ZZ} \int \dd \mu(U_{\Lambda\setminus\Gamma}) 
U_{\Lambda\setminus\Gamma} A U_{\Lambda\setminus\Gamma}^*
\label{eq:proj}
\end{equation}
where $\tr_X$ denotes a partial trace with respect to a local algebra supported on $X$. 
The RHS of (\ref{eq:proj}) provides a very useful identity, where $\dd\mu(U_X)$ denotes the normalized Haar measure for the integration over the full unitary group of a $N^{|X|}$ dimensional Hilbert space over sites $X$.

Bravyi, Hastings and Verstraete (BHV) \cite{BHV} have reformulated LRE (\ref{eq:LRE}) in a very convenient way, namely they have 
shown that (\ref{eq:LRE}) implies the estimate
\begin{equation}
\| \tau_t(f) - (\tau_t(f))_{\Gamma}\| \le \phi |X| \| f\| \exp{\left(-\mu ({\rm dist}(X,\ZZ\setminus\Gamma) - v |t|)\right)}
\label{eq:BHV}
\end{equation}
where $f \in \AAA_X$, and $\Gamma\subset\ZZ$ an arbitrary set of sites. The constants $\phi,\mu,v$ are the same as in (\ref{eq:LRE}) and hence also do not depend on $\Gamma$. The BHV inequality (\ref{eq:BHV}) is proven by applying the identity (\ref{eq:proj}), writing
\begin{eqnarray}
\| \tau_t(f) - (\tau_t(f))_{\Gamma} \| &=& \| \lim_{\Lambda\to\ZZ}\int \dd \mu(U_{\Lambda\setminus\Gamma}) 
[\tau_t(f),U_{\Lambda\setminus\Gamma}] U^*_{\Lambda\setminus\Gamma} \|
\nonumber \\&\le& \lim_{\Lambda\to\ZZ}\int \dd \mu(U_{\Lambda\setminus\Gamma}) 
\| [U_{\Lambda\setminus\Gamma},\tau_t(f)]\|,
\end{eqnarray} and then using LRE (\ref{eq:LRE}). Note that $U^*_{\Lambda\setminus\Gamma} = U^{-1}_{\Lambda\setminus\Gamma}$.

\section{Thermodynamic limit of Mazur inequality}
\label{sect:main}

Let us consider a suitable extensive translationally invariant observable $J$ (\ref{eq:current}).
Then we define the Drude weight in terms of a time-averaged autocorrelation function as the following double limit,
\begin{equation}
D_\beta:=\lim_{t\to \infty}\lim_{n\to \infty}\frac{\beta}{2n}\frac{1}{2t}\int_{-t}^{t} \dd t' 
\omega_{\beta}(J_{\Lambda_{n}}\tau_{t'}(J_{\Lambda_{n}})).
\label{eq:Dbeta}
\end{equation}
In order to avoid any ambiguity and to make the definition precise, we stress that the Gibbs state $\omega_\beta$ and the time-evolution $\tau_t$ on the RHS of (\ref{eq:Dbeta}) are already taken for an
infinite system, before the infinite volume limit $n\to\infty$ is applied to the extensive observable $J_{\Lambda_n}$. However, we will show later that RHS can be expressed in terms of local quantities only, therefore a single TL suffices.

The Drude weight $D_\beta$ is an important quantity in linear response theory of condensed matter physics (see the discussion in
subsection \ref{sect:lr}). In the following we prove a general and useful inequality related to it,
and clarify its existence in infinite one-dimensional systems. Moreover, $D_\beta$ can be considered as an interesting {\em ergodicity indicator} of $C^*$ dynamical systems.
 
\begin{theorem}
\label{theo:main}
(i) The following limit exists for any $t\in \RR$
\begin{equation}
c(t) = \lim_{n\to\infty} \sum_{x=-n}^n \omega_\beta( j_x \tau_t(j) ).
\label{eq:cdef1}
\end{equation}
(ii) Let us assume that the symmetric time-average of $c(t)$ exists
\begin{equation}
\bar{c} = \lim_{t\to \infty} \frac{1}{2t}\int_{-t}^t \dd s\, c(s),
\label{eq:cbar}
\end{equation}
together with a very mild condition, namely that the integral of $|t|(c(t) - \bar{c})$ grows slower than $t^2$
\begin{equation}
\lim_{t\to\infty}\frac{1}{t^2}\int_{-t}^t \dd s |s| (c(s) - \bar{c}) = 0.
\label{eq:ccbar}
\end{equation}
Then, the double-limit $D_\beta$ (\ref{eq:Dbeta}) exists and is equal to 
\begin{equation}
D_\beta = \frac{\beta}{2} \bar{c}.
\label{eq:Dc}
\end{equation}
(iii) Let $Q_{\Lambda_{n}}$ be a self-adjoint almost-con\-served quantity, satisfying {\em Definition~\ref{def:1}} (Eq. \ref{eqn:commutator}). 
Then, under the assumptions of (ii), a lower bound on $D_\beta$ exists  and is equal to the limit
\begin{equation}
D_\beta \geq \frac{\beta}{2}\lim_{n\to \infty}\frac{1}{n}
\frac{(\omega_{\beta}(\{J_{\Lambda_{n}},Q_{\Lambda_{n}}\}))^{2}}{4\omega_{\beta}(Q^{2}_{\Lambda_{n}})}.
\label{eqn:theorem}
\end{equation}
\end{theorem}
$\{A,B\} \equiv AB+BA$ denotes the anti-commutator.
\begin{proof}
We start by considering the following finite-time-averaged self-adjoint operator
\begin{equation}
A_{\Lambda_{n},t}:=\frac{1}{\sqrt{n}}\frac{1}{t}\int_{0}^{t}\dd t'(\tau_{t'}(J_{\Lambda_{n}})-\alpha Q_{\Lambda_{n}}),
\label{eqn:ansatz}
\end{equation}
where  $\alpha \in \mathbb{R}$ is a free parameter.
Since the state $\omega_{\beta}$ is a positive linear functional, we have $\omega_{\beta}(A^{2}_{\Lambda_{n},t})\geq 0$ for any $t,\alpha \in \mathbb{R}$, $n\in\ZZ_+$, or equivalently:
\begin{eqnarray}
0 &\leq& \frac{1}{t^{2}}\int_{0}^{t}\dd t'\int_{0}^{t}\dd t'' \frac{1}{n} \omega_{\beta}(\tau_{t'}(J_{\Lambda_{n}})\tau_{t''}(J_{\Lambda_{n}})) \label{eq:AC} \\
&-& \frac{\alpha}{t} \int_0^t \dd t' \frac{1}{n}\left\{\omega_{\beta}(\tau_{t'}(J_{\Lambda_{n}})Q_{\Lambda_{n}})
+ \omega_{\beta}(Q_{\Lambda_{n}}\tau_{t'}(J_{\Lambda_{n}}))\right\}  \label{eq:mixedterms}\\
&+&\frac{\alpha^{2}}{n}\omega_{\beta}(Q^{2}_{\Lambda_{n}}). \label{eq:QQ}
\end{eqnarray}
We shall proceed to show that TL  $n\to\infty$ of all the terms in the above inequality exists, treating it term by term.

Let us first discuss in detail the mixed terms (\ref{eq:mixedterms}). 
Time invariance of the Gibbs state (\ref{eq:translt}) implies that the integrand of (\ref{eq:mixedterms}) can be rewritten as (writing now time integration variable as $t$)
\begin{equation}
\omega_{\beta}(J_{\Lambda_{n}}\tau_{-t}(Q_{\Lambda_{n}})) + \omega_{\beta}(\tau_{-t}(Q_{\Lambda_{n}})J_{\Lambda_{n}}).
\end{equation}
As both terms can be treated on equal footing, we shall focus on the first one 
$\omega_{\beta}(J_{\Lambda_{n}}\tau_{-t}(Q_{\Lambda_{n}}))$ and show that in the limit $n\to\infty$ the term becomes time-independent.
Using the explicit time evolution (\ref{eq:explicit}), the linearity of the state $\omega_\beta$, and an elementary integral triangular inequality, we estimate
\begin{eqnarray}
&&\left|  \omega_{\beta}(J_{\Lambda_{n}}\tau_{-t}(Q_{\Lambda_{n}})) -  \omega_{\beta}(J_{\Lambda_{n}} Q_{\Lambda_{n}}) \right| \nonumber \\
&&\le \int_{-t}^0\dd s\, \left|\omega_{\beta}(J_{\Lambda_{n}} \tau_{s}(B_{\rm L}+B_{\rm R}))\right| \nonumber\\
&&
\le \int_{-t}^0\dd s\, \left(
\left|\omega_{\beta}(J_{\Lambda_{n}} \tau_{s}(B_{\rm L}))\right| +
\left|\omega_{\beta}(J_{\Lambda_{n}} \tau_{s}(B_{\rm R}))\right|\right).
\label{eq:bound}
\end{eqnarray}
We note that, before taking the TL $n\to\infty$, the mixed terms (\ref{eq:mixedterms}) have to be multiplied by $1/n$, hence it is enough to show that all terms on the RHS of (\ref{eq:bound}), for arbitrary time variables $s$, {\em do not grow with} $n$, i.e. are bounded by constants that do not depend on $n$.
Again, it is enough to focus just on one of the terms in the integrand, say $|\omega_\beta(J_{\Lambda_n} \tau_s(B_{\rm L}))|$, while the other one can be treated in exactly the same way.
Expressing the operators in terms of local densities, we obtain:
\begin{figure}
\centering	
\includegraphics[width=1\columnwidth]{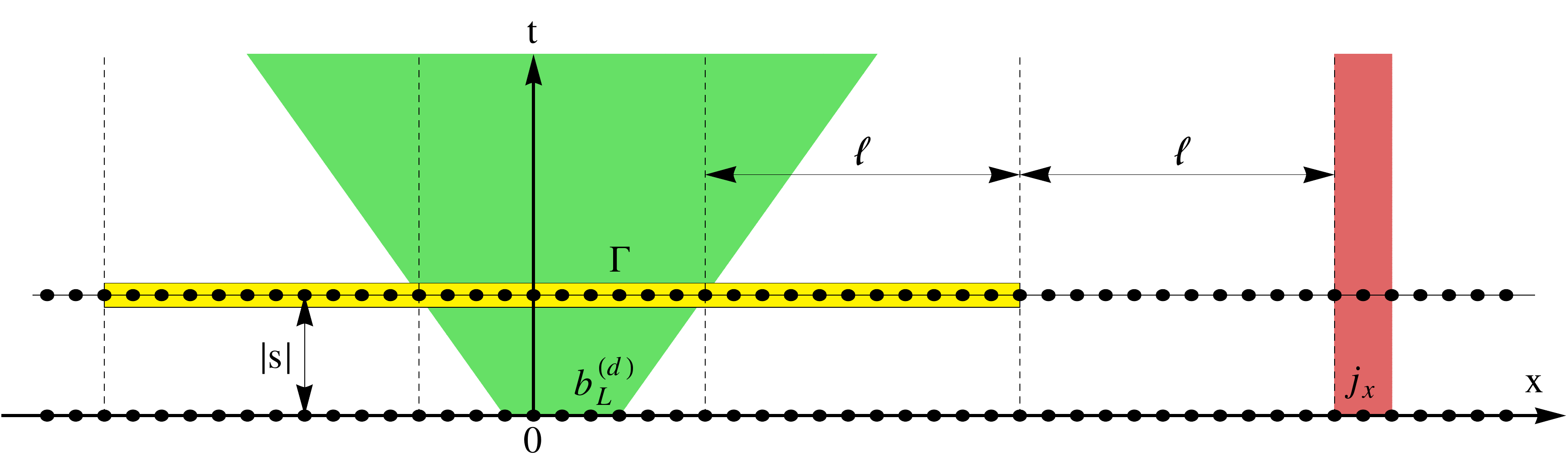}
\caption{A scheme visualizing the estimation of the crucial term $|\omega_\beta(j_x \tau_s(b^{(d)}_{\rm L}))|$.
The operator $\tau_s(b^{(d)}_{\rm L})$, supported near the left boundary at time $s=0$, is localized within an effective
light-cone (light-green region) and can be approximated with an error exponentially small in half-gap $\ell$ (indicated) to the support of $j_x$ (red strip), by projecting it onto a sublattice $\Gamma$ (painted in yellow).}
\label{fig:lightcones}
\end{figure}

\begin{equation}
|\omega_\beta(J_{\Lambda_n} \tau_s(B_{\rm L}))| \le 
\sum_{d=0}^n \sum_{x=1}^{n-d_j+1}\left|\omega_\beta(j_x \tau_s(b^{(d)}_{\rm L}))\right|.
\label{eq:xdsum}
\end{equation}
For the time being, let us fix the integration time variable $s$, and the summation variables $d,x$.
We first assume that the Lieb-Robinson light-cone emitting from the local operator $b_{\rm L}^{(d)}$, at time $s$, is not including $x$, i.e. there is a positive half-gap $\ell > 0$,
defined as (see Fig.~\ref{fig:lightcones}):
\begin{equation}
\ell = \frac{1}{2}\left( x - v|s| - \max\{d_h-2+d,d_b\}\right).
\end{equation}
Defining the chain of sites including the light-cone plus the half-gap of sites on each side as
\begin{equation}
\Gamma= \left[-d_h+2 - \lfloor v|s| + \ell \rfloor, \max\{d_h-2+d,d_b\} + \lfloor v|s| + \ell \rfloor\right] 
\label{eq:Gamma}
\end{equation}
we can estimate the term on the RHS of (\ref{eq:xdsum}) as
\begin{eqnarray}
\!\!\!\!\!\!|\omega_\beta(j_x \tau_s(b^{(d)}_{\rm L}))| &\le& 
|\omega_{\beta}(j_x (\tau_s(b^{(d)}_{\rm L}))_\Gamma)| + |\omega_{\beta}(j_x (\tau_s(b^{(d)}_{\rm L})-(\tau_s(b^{(d)}_{\rm L}))_\Gamma))|,
\end{eqnarray}
where the first term is further estimated using the ECP (\ref{eq:ECP}), also noting (\ref{eq:jzero}),
\begin{eqnarray}
|\omega_{\beta}(j_x (\tau_s(b^{(d)}_{\rm L}))_\Gamma)| &=& |\omega_{\beta}(j_x (\tau_s(b^{(d)}_{\rm L}))_\Gamma)-\omega_\beta(j_x)\omega_\beta((\tau_s(b^{(d)}_{\rm L}))_\Gamma)| 
\nonumber\\
&\le& \kappa  \| j\| \|   (\tau_s(b^{(d)}_{\rm L}))_\Gamma\| e^{-\rho(\ell-1)}, \label{eq:ECPapp}
\end{eqnarray}
while the second term is bounded by BHV form of the LRE (\ref{eq:BHV})
\begin{eqnarray}
&&|\omega_{\beta}(j_x (\tau_s(b^{(d)}_{\rm L})-(\tau_s(b^{(d)}_{\rm L}))_\Gamma))|
\le \| j\| \| \tau_s(b^{(d)}_{\rm L})-(\tau_s(b^{(d)}_{\rm L}))_\Gamma\| \label{eq:BHVapp} \\
&&\le \phi \max\{d_b+d_h-1,2 d_h-3+d\} \max\{ d_h(d_h-1) \gamma \| h\| ,\| b\|\} \| j \| e^{-\mu(\ell -1)-\xi d}. \nonumber
\end{eqnarray}
The norm of projected evolution can be estimated generously by another use of the triangular and the BHV inequalities
\begin{eqnarray}
\| (\tau_s(b^{(d)}_{\rm L}))_\Gamma \| &\le&
\| \tau_s(b^{(d)}_{\rm L}) \| + \| \tau_s(b^{(d)}_{\rm L}) - (\tau_s(b^{(d)}_{\rm L}))_\Gamma\| \label{eq:projnorm}\\
&\le& (1 + \phi \max\{d_b+d_h-1,2 d_h-3+d\}) \nonumber \\
&& \times \max\{d_h (d_h-1) \gamma \| h\|,\| b\|\} e^{-\xi d}.  \nonumber
\end{eqnarray}
Putting Eqs. (\ref{eq:ECPapp},\ref{eq:BHVapp},\ref{eq:projnorm}) together, writing the minimal exponent $\lambda = \min\{\rho,\mu\} > 0$, and taking suitable constants
$C,C' > 0$, which only depend on the local operator norms and dimensions but {\em not} on the size $n$ {\em neither} on the variables $x$, $d$, and $s$, we have:
\begin{equation}
|\omega_\beta(j_x \tau_s(b^{(d)}_{\rm L}))| \le (C d + C') e^{-\lambda \max\{0,\ell\}-\xi d }.
\label{eq:exponential}
\end{equation} 
For the bound to remain valid for an overlapping light-cone, $\ell < 0$, we should simply make sure that we choose $C$ and $C'$ large enough to satisfy the naive bound
$|\omega_\beta(j_x \tau_s(b^{(d)}_{\rm L}))| \le \| j\| \| b^{(d)}_{\rm L}\|$.
The whole term (\ref{eq:xdsum}) can now be estimated as, introducing $k(d) := \lfloor v|s| + \max\{d_h-2+d,d_b\}\rfloor$:
\begin{eqnarray}
&&\sum_{d=0}^n \sum_{x=1}^{n-d_j+1}\left|\omega_\beta(j_x \tau_s(b^{(d)}_{\rm L}))\right| \nonumber \\
&&\le \sum_{d=0}^n (C d + C') e^{-\xi d} \left(k(d) + 
\sum_{x'=0}^{\max\{0,n-k(d)\}} e^{-\lambda x'/2}\right) \nonumber \\
&&
\le \sum_{d=0}^\infty (C d + C') e^{-\xi d} \left( v|s| + d_h+d_b-2 + d + \sum_{x'=0}^\infty e^{-\lambda x'/2}\right) \nonumber\\
&&= \sum_{d=0}^\infty (C d^2 + C'' d + C''') e^{-\xi d}\nonumber\\ 
&&= \frac{e^\xi ((e^\xi+1)C + (e^\xi-1)(C''+(e^\xi-1)C'''))}{(e^\xi-1)^3} =: K < \infty,
\end{eqnarray}
where $C''=C' + C (v|s| + d_h + d_b - 2 + (1-e^{-\lambda/2})^{-1})$ and $C'''=C'  (v|s| + d_h + d_b - 2 + (1-e^{-\lambda/2})^{-1})$.
As $C''$ and $C'''$ are at most linear functions of time $s$, we have also that $K = K' + K'' |s|$ where $K'$ and $K''$ are
$n$ and $s$ independent constants. Exactly the same estimates applies for the other term $|\omega_{\beta}(J_{\Lambda_{n}} \tau_{s}(B_{\rm R}))|$,
so we have finally shown that the difference on the left-hand-side (LHS) of (\ref{eq:bound}) is bounded by
\begin{equation}
\left|  \omega_{\beta}(J_{\Lambda_{n}}\tau_{-t}(Q_{\Lambda_{n}})) -  \omega_{\beta}(J_{\Lambda_{n}} Q_{\Lambda_{n}}) \right| 
\le 2 K' |t| + K'' t^2.
\label{eq:KK}
\end{equation}

We shall now study the convergence properties of the sequences 
\begin{eqnarray}
w_n &=& \frac{1}{n}\omega_\beta(J_{\Lambda_n} Q_{\Lambda_n}),\nonumber\\
w'_n &=& \frac{1}{n}\omega_\beta(Q_{\Lambda_n}J_{\Lambda_n} ),\\
u_n &=& \frac{1}{n}\omega_\beta(Q_{\Lambda_n}^2), \nonumber
\end{eqnarray} 
and show that they are, in fact, Cauchy sequences.

Let us consider, for the time being, an abstract sequence of this type
\begin{equation}
v_n = \frac{1}{n}\omega_\beta(F_{\Lambda_n} G_{\Lambda_n}),
\end{equation} 
where 
\begin{equation}
F_{\Lambda_n}=\sum_{x=1}^{n-d_f+1} f_x, \quad G_{\Lambda_n}=\sum_{x=1}^{n-d_g+1} g_x, \quad f \in \AAA_{[0,d_f-1]},\quad  g \in \AAA_{[0,d_g-1]}, 
\end{equation}
with $\omega_\beta(f)=\omega_\beta(g)=0$.
Exploiting the translational invariance of the Gibbs state $\omega_\beta(f_x g_y) = \omega_\beta(f g_{y-x})$ we find (see Fig.~\ref{fig:congo})
\begin{eqnarray}
v_n &=& \frac{1}{n}\sum_{x=1}^{n-d_f+1} \sum_{y=1}^{n-d_g+1} \omega_\beta(f_x g_y) \nonumber \\
&=& \left(\textstyle{1-\frac{\max\{d_f,d_g\}-1}{n}}\right)\sum_{r=\min\{0,d_f-d_g\}}^{\max\{0,d_f-d_g\}} \omega_\beta(f g_r) \nonumber \\
&&+ \sum_{r=\max\{1,d_f-d_g+1\}}^{n-d_g}\left(\textstyle{1 - \frac{d_g-1+r}{n}}\right)\omega_\beta(f g_r)  \nonumber \\
&&+ \sum_{r=-n+d_f}^{\min\{-1,d_f-d_g-1\}} \left(\textstyle{1 - \frac{d_f-1-r}{n}}\right)\omega_\beta(f g_r).  \label{eq:mm} 
\end{eqnarray}
Clearly, ECP (\ref{eq:ECP}), combined with the trivial norm-bound,
\begin{equation}
|\omega(f g_r)| \le \|f\| \|g\| \min\{1,\kappa\exp(-\rho(r-d_f+1)),\kappa\exp(-\rho(-r-d_g+1))\} 
\label{eq:ECPx}
\end{equation}
guarantees boundedness of the sequence $|v_n| \le V < \infty$.
But we have more:
\begin{eqnarray}
&&|v_{n+1}-v_{n}| = \frac{1}{n(n+1)}\Biggl|
(\max\{d_f,d_g\}-1)\!\!\!\!\!\!\!\sum_{r=\min\{0,d_f-d_g\}}^{\max\{0,d_f-d_g\}}\!\!\!\!\!\!\!\omega_\beta(f g_r) \nonumber \\
&&+\!\!\!\!\!\!\!\sum_{r=\max\{1,d_f-d_g+1\}}^{n+1-d_g}\!\!\!\!\!\!\!(d_g-1+r)\omega_\beta(f g_r)+\!\!\!\!\!\!\!\sum_{r=-n-1+d_f}^{\min\{-1,d_f-d_g-1\}}\!\!\!\!\!\!\!(d_f-1-r)\omega_\beta(f g_r)
\Biggr| \nonumber \\
&& 
\le \frac{\| f\| \|g\|}{n(n+1)} \Bigl\{ (\max\{d_f,d_g\}-1)(|d_f-d_g|+1) \nonumber \\
&&\quad + \!\!\!\!\!\sum_{r=\max\{1,d_f-d_g+1\}}^{d_f-1}\!\!\!\!\!\!\!(d_g-1+r) + \kappa\!\sum_{r=d_f}^{n+1-d_g} (d_g-1+r)e^{-\rho(r-d_f+1)} \nonumber \\
&&\quad + \!\!\sum_{r=-d_g+1}^{\min\{-1,d_f-d_g-1\}}\!\!\!\!\!\!\!(d_f-1-r) + \kappa\!\!\!\!\!\!\sum_{r=-n-1+d_f}^{-d_g}\!\!\!\!\! (d_f - 1 - r) e^{-\rho(-r-d_g+1)}  
\Bigr\} \label{eq:50}\\
&&= \frac{\| f\| \|g\|}{n(n+1)} \Bigl\{ d_f^2+d_g^2+d_f d_g-2 d_f-2 d_g+1 + 2\kappa\!\!\!\!\!\!\!\sum_{r=1}^{n+1-d_f-d_g}\!\!\!\!\!\!\!\!(r\!+\!d_f\!+\!d_g\!-\!2) e^{-\rho r}\Bigr\}  \nonumber
\end{eqnarray}
where we have adopted a convention that $\sum_{r=x}^y(\ldots) = 0$ if $x > y$.
Finally, we complete the sum in the curly bracket to a geometric series and simplify the denominator to arrive at the quickly decreasing bound on the difference between the adjacent terms
\begin{equation}
|v_{n+1}-v_{n}| 
\le \frac{1}{n^2} \| f\| \|g \| \bigl\{ d_f^2+d_g^2+d_f d_g + \nu(d_f + d_g-1) + \zeta\bigr\},
\end{equation}
where $\nu:=2({\textstyle \frac{\kappa}{e^\rho-1}}-1)$, $\zeta:=\frac{2\kappa}{(e^\rho-1)^{2}}-1$,
which proves that $\{ v_n\}$ is a Cauchy sequence, i.e. $\lim_{n\to \infty} v_n$ exists and is finite. In fact it is equal to the limit
\begin{equation}
\lim_{n\to\infty} v_n = \lim_{n\to\infty} \sum_{x=-n}^n \omega_\beta(f g_x),
\label{eq:fglimit}
\end{equation}
since it can be shown -- again using (\ref{eq:mm},\ref{eq:ECPx}) -- that the difference of the terms of the sequences on the LHS and the RHS of (\ref{eq:fglimit}) is bounded by
${\rm const}/n$.

\begin{figure}
\centering
\includegraphics[width=0.7\columnwidth]{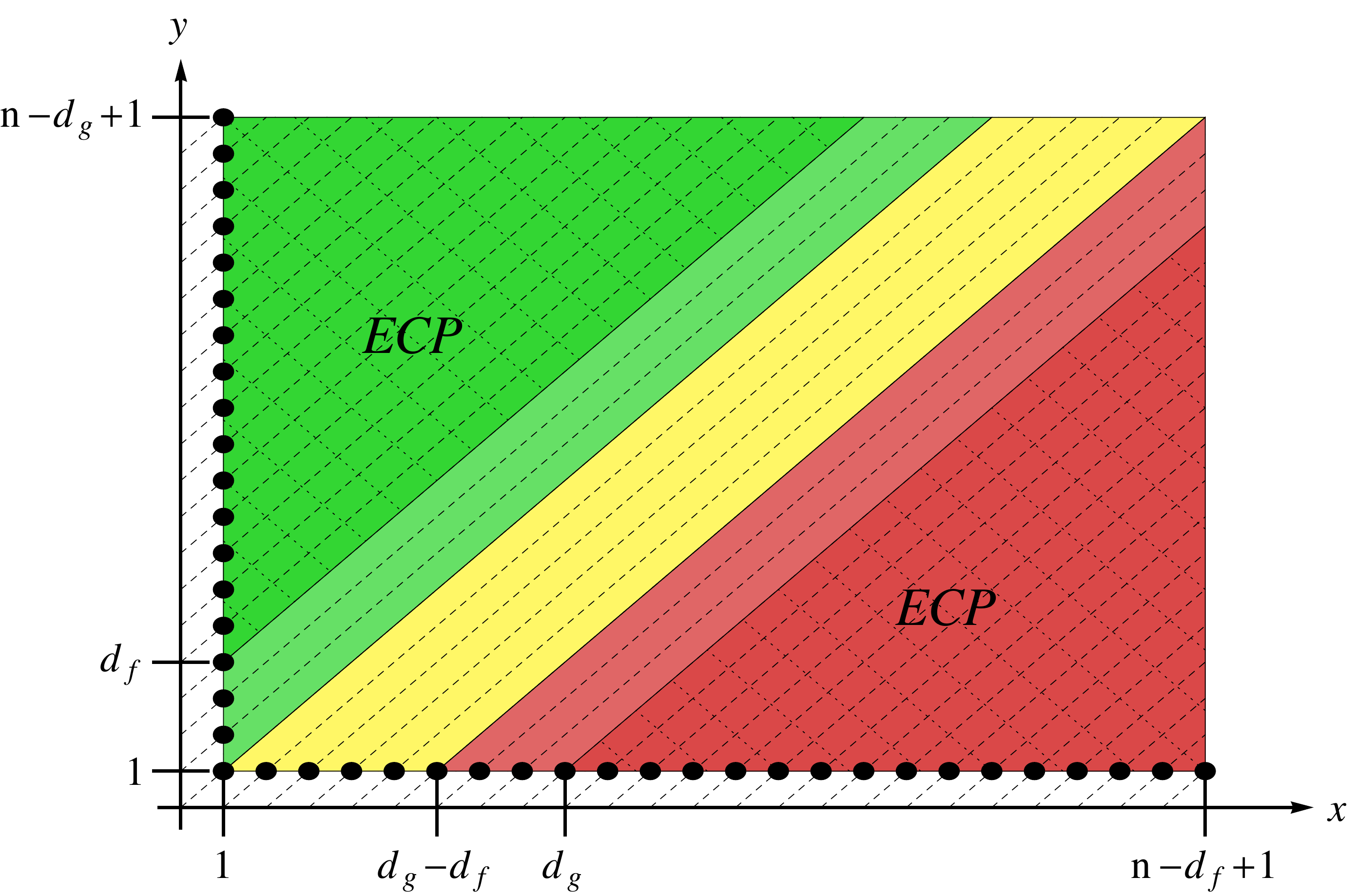}
\caption{A scheme visualizing the double summation over $x$ and $y$ for computing the sequence $v_n$ (\ref{eq:mm}).
Yellow, green, and red region, represent the terms under the first, the second, and the third sum on the RHS of (\ref{eq:mm}), respectively.
Light-green and light-red stripes denote the terms which are separated out as the first (left) summations in the last two lines of expression (\ref{eq:50}) before ECP (\ref{eq:ECP}) is applied to the remaining terms.}
\label{fig:congo}
\end{figure}

Similarly we now treat the sequence $w_n = n^{-1} \omega_\beta(J_{\Lambda_n} Q_{\Lambda_n})$. Writing $w_n=n^{-1}\sum_{d=1}^n \omega_{\beta}(J_{\Lambda_n} Q^{(d)}_{\Lambda_n})$,
and associating $f\equiv j$, $g\equiv q^{(d)}$, we arrive at (again in the last step completing the geometric series)
\begin{eqnarray}
&&|w_{n+1}-w_n| \le \frac{\gamma \| j\|}{n^2}
\sum_{d=1}^n e^{-\xi d} (d^2 + (d_j + \nu) d + d_j^2 + \nu (d_j-1)+\zeta
) \nonumber \\
&&\le \frac{\gamma\| j\|}{n^2} \frac{2 + (e^\xi-1)(d_j+3+\nu)+(e^\xi-1)^2((d_j+1+\nu)d_j+1+\zeta)}{(e^\xi-1)^3},
\end{eqnarray}
and with the same bound for the `transposed' sequence $w'_n = n^{-1}\omega_\beta(Q_{\Lambda_n} J_{\Lambda_n})$.
Analogously, for the sequence  $u_n = n^{-1}\omega_\beta(Q_{\Lambda_n} Q_{\Lambda_n})=n^{-1}\sum_{d,d'=1}^n \omega_{\beta}(Q^{(d)}_{\Lambda_n} Q^{(d')}_{\Lambda_n})$,
we obtain
\begin{eqnarray}
|u_{n+1}-u_n| &\le& \frac{\gamma^2}{n^2}
\sum_{d,d'=1}^n e^{-\xi (d+d')} (d^2+{d'}^2+d d' + \nu (d+d'-1)+\zeta) \nonumber \\
&\le& \frac{\gamma^2}{n^2} \frac{\zeta - \nu + e^\xi (2 - 2\zeta + e^\xi(\zeta + \nu + 3))}{(e^\xi-1)^4}.
\end{eqnarray}
We have thus shown that the limits 
\begin{eqnarray}
w&=&\lim_{n\to\infty}w_n = \sum_{d=1}^\infty \sum_{x\in\ZZ} \omega_\beta(j q^{(d)}_x), \nonumber \\
w'&=&\lim_{n\to\infty} w'_n = \sum_{d=1}^\infty \sum_{x\in\ZZ} \omega_\beta(q^{(d)} j_x),\\
u&=&\lim_{n\to\infty} u_n = \sum_{d,d'=1}^\infty \sum_{x\in\ZZ} \omega_\beta(q^{(d)}q^{(d')}_x), \nonumber
\end{eqnarray} 
exist and are finite (see Eq. (\ref{eq:fglimit}) for the RHS).
Moreover, Eq. (\ref{eq:KK}) implies that the distance between $w_n+w'_n$ and the integrand in (\ref{eq:mixedterms}),
\begin{equation}
z_n(t):= n^{-1}\left\{\omega_{\beta}(\tau_{t}(J_{\Lambda_{n}})Q_{\Lambda_{n}})
+ \omega_{\beta}(Q_{\Lambda_{n}}\tau_{t}(J_{\Lambda_{n}}))\right\}
\end{equation} 
is decreasing in $n$
\begin{equation}
\left| w_n+w'_n-z_n(t) \right| \le \frac{ 4 K'|t| + 2 K'' t^2}{n},
\end{equation} 
so the two sequences should have the same limit, i.e. $\lim_{n\to\infty}z_n(t)=w+w'$. Therefore, the
TL of (\ref{eq:mixedterms}) exists, it is independent of $t$, and equals $w + w'$. 
  
At last, let us devote also some attention to the first term (\ref{eq:AC}), which after exploiting the time-invariance of the Gibbs state and substituting for  the integration variable $s=t''-t'$, reads:
\begin{equation}
a_n(t) = \frac{1}{t}\int_{-t}^t \dd s \left(1 - \frac{|s|}{t}\right) c_n(s),\
\label{eq:an}
\end{equation}
where 
\begin{equation}
c_n(t) := \frac{1}{n}\omega_\beta(J_{\Lambda_n}\tau_t(J_{\Lambda_n})) = \sum_{r=-n+d_j}^{n-d_j}{\textstyle\left(1-\frac{|r|+d_j-1}{n}\right)}\omega_\beta(j_r \tau_t(j)).
\label{eq:cn}
\end{equation}
For each fixed time $t$, we can again show that $\{ c_n(t)\}$ is a Cauchy sequence, namely
\begin{equation}
|c_{n+1}(t)-c_n(t)| \le \frac{1}{n(n+1)}\sum_{r=-\infty}^\infty (|r|+d_j-1) |\omega_\beta(j_r \tau_t(j))| < \frac{L|t|+L'}{n^2} 
\end{equation}
where the constants $L,L'$ follow from bounding the sum of exponentially decaying (in $x$)
envelope of the spatio-temporal autocorrelation function
\begin{equation}
|\omega_\beta(j_x \tau_t(j))| \le \| j\|^2 \min\{1, e^{-\lambda(|x|-d_j-v|t|)/2}\}.
\label{eq:expon}
\end{equation}
Eq. (\ref{eq:expon}) can be proven in exactly the same way as Eq.(\ref{eq:exponential}) -- combining the BHV inequality and the ECP -- simply taking $j$ to substitute the boundary operator $b^{(d)}_{\rm L}$.
This implies that the limit of $c_n(t)$ exists. It also coincides with the expression (\ref{eq:cdef1}) of $c(t)$, as the difference between the terms of sequences (\ref{eq:cn}) and (\ref{eq:cdef1}) can be bounded by ${\rm const}/n$ (again using (\ref{eq:expon})). 

This proves the point (i) of the theorem.

As $c_n(t)$ converges uniformly in $t$, the limit $a(t)=\lim_{n\to\infty} a_n(t)$ of (\ref{eq:an}) exists as well and is equal to
\begin{equation}
a(t) = \frac{1}{t}\int_{-t}^t \dd s \left(1-\frac{|s|}{t}\right) c(s).
\end{equation}
Similarly, a uniform convergence of $c_n(t)$ can be used to express the Drude weight (\ref{eq:Dbeta}) as
\begin{equation}
D_\beta = \frac{\beta}{2} \lim_{t\to\infty}\frac{1}{2t}\lim_{n\to\infty} \int_{-t}^t \dd s\,c_n(s) = \frac{\beta}{2} \lim_{t\to\infty}\frac{1}{2t}\int_{-t}^t \dd s\,c(s) = \frac{\beta}{2}\bar{c}
\end{equation}
proving the point (ii) of the theorem.

Finally, writing the TL $n\to\infty$ of all the terms in the inequality (\ref{eq:AC}-\ref{eq:QQ})
we have
\begin{equation}
0 \le a(t) - \alpha (w + w') + \alpha^2 u.
\label{eq:at}
\end{equation}
One should observe an obvious identity
\begin{equation}
a(t) - \bar{c} = 2 \left(\frac{1}{2t}\int_{-t}^t \dd s\, c(s) - \bar{c}\right) - \frac{1}{t^2} \int_{-t}^t \dd s |s|(c(s) - \bar{c}),
\end{equation}
where, according to the assumptions (\ref{eq:cbar}) and (\ref{eq:ccbar}), the limits $t\to\infty$ of both terms on the RHS exist and vanish, hence $\lim_{t\to\infty} a(t) = \bar{c}$.
Taking then the limit $t\to\infty$ of the inequality (\ref{eq:at}) and optimizing it with respect to the free parameter $\alpha$,
we arrive to the final bound
\begin{equation}
\bar{c} \ge \frac{(w+w')^2}{4 u}.
\end{equation}
This proves the point (iii) of the theorem.
\end{proof}
Let us conclude this section by making a few remarks.
\begin{remark}
The key in our proof was a convenient form of the ansatz (\ref{eqn:ansatz}). It  has been inspired by combining a finite-time-average version of the original idea of Mazur 
\cite{Mazur} with an expected central limit theorem behavior in size $n$ for local spin chains (hence the $1/\sqrt{n}$ prefactor) which allows taking the TL $n\to\infty$ first.
\end{remark}
\begin{remark}
It is not clear at present if in the context of quasi-local $C^*$ dynamical systems, the assumptions of point (ii) of the Theorem 1, e.g. on the existence of the time average $\bar{c}$, are in fact needed, or can be separately proven. In cases of more general dynamics it is of course easy to come up with counterexamples of dynamics for which $\bar{c}$ does not exist. For example, one may formally construct a gaussian process with the (2-point) correlator $c(t) = (-1)^{\lfloor \log_2 (t/t_0)\rfloor}$ for which, clearly, the average $\bar{c}$ does not exist.
\end{remark}
\begin{remark}
Note that the RHS of our estimate (\ref{eqn:theorem}) can in fact vanish in some cases, e.g. when the limits $w,w'$ vanish. In such cases our bound is nearly trivial, and merely expresses the fact that the Drude weight as defined by (\ref{eq:Dbeta}) is always non-negative.
\end{remark}

\section{Generalization to the case of several almost-conserved quantities}
\label{sect:general}

The bound (\ref{eqn:theorem}) of the Theorem~\ref{theo:main} can be readily generalized to the case where one has several almost-conserved quantities.

\begin{theorem}
\label{theo:general}
Let $Q_{[k]\Lambda_{n}}$, $k=1,\ldots,m$ be a set of $m$ self-adjoint, almost-con\-served quantities with densities $q^{(d)}_{[k]}\in \AAA_{[0,d-1]}$, all satisfying {\em Definition~\ref{def:1}} (Eq. \ref{eqn:commutator}).
Then, under the assumptions of the Theorem 1, together with the points (i) and (ii) of Theorem 1 we have the following lower bound on $D_\beta$ (\ref{eq:Dbeta})
\begin{equation}
D_\beta \ge \frac{\beta}{2} \sum_{k,l=1}^m w_k (\mm{U}^{-1})_{k,l} w_l
\label{eq:generalbound}
\end{equation}
where
\begin{equation}
w_k = \frac{1}{2} \lim_{n\to\infty} \omega_\beta(\{J_{\Lambda_n},Q_{[k]\Lambda_n}\}) = \frac{1}{2}\sum_{d=1}^\infty\sum_{x\in\ZZ} 
\omega_\beta(\{ j_x, q^{(d)}_{[k]}\})
\label{eq:wk}
\end{equation}
and $\mm{U}^{-1}$ is an inverse of $m\times m$ covariance matrix
\begin{equation}
\mm{U}_{k,l} = \frac{1}{2} \lim_{n\to\infty} \omega_\beta(\{Q_{[k]\Lambda_n},Q_{[l]\Lambda_n}\}) = \frac{1}{2}\sum_{d,d'=1}^\infty\sum_{x\in\ZZ} 
\omega_\beta(\{\eta_x(q^{(d)}_{[k]}),q^{(d')}_{[l]}\}).
\label{eq:Ukl}
\end{equation}
\end{theorem}
\begin{proof}
The proof involves exactly the same steps as the proof of the Theorem~\ref{theo:main}, except that the ansatz (\ref{eqn:ansatz}) is replaced by a more general one
\begin{equation}
A_{\Lambda_{n},t}:=\frac{1}{\sqrt{n}}\frac{1}{t}\int_{0}^{t}\dd t'\left(\tau_{t'}(J_{\Lambda_{n}})-\sum_{k=1}^m\alpha_{k}Q_{[k]\Lambda_{n}}\right),
\end{equation}
with $m$ real free parameters $\alpha_k$.
The final step of optimization of $m$-dimensional quadratic form (in $\alpha_k$) then results in (\ref{eq:generalbound}).
The limit-identity (\ref{eq:fglimit}) is used to write the compact expressions on the RHSs of (\ref{eq:wk},\ref{eq:Ukl}).
\end{proof}

\section{Examples}
\label{sect:examples}

As a pool of nontrivial examples let us discuss the anisotropic Heisenberg spin 1/2 chain (the so-called $XXZ$ model).
Here $N=2$, and the local algebra $\AAA_{[0]}$ is spanned by Pauli matrices $\sigma^{s}$, $s\in\{0,\x,\y,\z\}$, with $\sigma^0 = \one$, or $\sigma^{\pm} = \frac{1}{2}(\sigma^{\rm x} \pm \ii \sigma^{\rm y})$.
The Hamiltonian density reads
$h = \sigma^\x_0 \sigma^\x_1 + \sigma^\y_0 \sigma^\y_1 + \Delta \sigma^\z_0 \sigma^\z_1$, $d_h=2$, where $\Delta$ is the anisotropy parameter.

Grabowski and Mathieu \cite{Grabowski} have shown that an infinite sequence of nontrivial local conservation laws of an infinite $XXZ$ chain can be formally constructed using the {\em boost operator}
$B= \frac{1}{2}\sum_{x\in\ZZ} x h_x$, namely 
\begin{equation}
Q_{[k+1]} = [B,Q_{[k]}], \; k = 2,3,\ldots,\qquad Q_{[2]} = \sum_{x\in\ZZ} h_x,
\label{eq:Qk}
\end{equation}
supplemented with the trivial on-site conservation law - the total magnetization - $Q_{[1]} \equiv M = \sum_{x\in\ZZ} \sigma^\z_x$.

Translationally invariant operators (\ref{eq:Qk}) which can clearly be written in terms of $k-$site densities $Q_{[k]}=\sum_{x\in\ZZ} \eta_x(q_{[k]})$, $q_{[k]}\in\AAA_{[0,k-1]}$,
are strictly conserved, i.e. they exactly commute with the Hamiltonian $Q_{[2]}$ and among each other $[Q_{[k]},Q_{[l]}] = 0$, {\em only} for infinite lattice ($\ZZ$) or for periodic boundary conditions.
Note however, that they {\em do not exist} as elements of $C^*$-algebra $\AAA$ in the former case, while in the latter case the `limit by inclusion' construction $\Lambda \to\ZZ$ does not work.

In the setup of $C^*$-algebraic statistical mechanics, the operators (\ref{eq:Qk}) should be considered on a finite open chain $\Lambda_n$, written as $Q_{[k]\Lambda_n}$, being the elements of $\AAA$.
Interestingly, it has been shown again by Grabowski and Mathieu \cite{Grabowski2} that for open boundary conditions, i.e. considering the Hamiltonian $H_{\Lambda_n}$, half of the conservation laws (for {\em odd} $k$) are destroyed,
while the other half (for {\em even} $k$) can be amended by adding terms supported near the boundary of $\Lambda_n$: $Q'_{[2l]\Lambda_n} = Q_{[2l]\Lambda_n} + Q_{[2l]{\rm L}} + Q_{[2l]{\rm R}}$, such that $[Q'_{[2l]\Lambda_n},Q'_{[2l']\Lambda_n}]=0$,
for $l,l'=1,\ldots,\lfloor n/2\rfloor$.
Nevertheless, as we have shown in the sect.~\ref{sect:main},\ref{sect:general}, one does not need {\em exact} conservation to state dynamic susceptibility bounds (\ref{eqn:theorem},\ref{eq:generalbound}).
Therefore, we rewrite the procedure (\ref{eq:Qk}) of \cite{Grabowski} in terms of densities $q_{[k]} \in \AAA_{[0,k-1]}$ and boundary remainders $p_{[k]} \in \AAA_{[0,k]}$
such that one has {\em almost-commutation relations}
\begin{equation}
\ii [H_{\Lambda_n},Q_{[k]\Lambda_n}] = \eta_1(p_{[k]}) - \eta_{n-k}(p_{[k]}).
\end{equation}
Operators $p_{[k]}$, which can be interpreted as the current densities, are actually determined from the local continuity equations $(\dd/\dd t) q_{[k]} = \eta_{-1}(p_{[k]}) - p_{[k]}$, or
\begin{equation}
p_{[k]}-\eta_1(p_{[k]}) = \ii \sum_{x=0}^k [h_x,\eta_1(q_{[k]})].
\label{eq:rec1}
\end{equation}
Local-algebraic version of the boost relation (\ref{eq:Qk}) gives the other recurrence relation which determines the charge density in the next order in terms of the charge and current densities of the previous order:
\begin{equation}
q_{[k+1]} = \frac{1}{2}p_{[k]} + \frac{\ii}{2}\sum_{x=0}^{k-1}(x + 1)[h_x,q_{[k]}], \quad k=2,3\ldots
\label{eq:rec2}
\end{equation}
Clearly, $p_{[1]} = j=2(\sigma^\x_0 \sigma^\y_1 - \sigma^\y_0\sigma^\x_1)$ is just the spin current, $q_{[2]}=h$, while the first few higher charges and densities can easily be obtained solving the recurrence (\ref{eq:rec1},\ref{eq:rec2}) by means of some computer algebra
\footnote{Using the {\em Mathematica} code {\tt http://chaos.fmf.uni-lj.si/prosen?action=AttachFile\&do\\=get\&target=PauliAlgebra.nb} one may obtain explicit form of $q_{[k]}$ and $p_{[k-1]}$ up to $k=10$.}
\begin{eqnarray}
q_{[3]} &=&-\Delta  \sigma^{\x\y\z}+\Delta  \sigma^{\y\x\z}-\Delta  \sigma^{\z\x\y}+\Delta  \sigma^{\z\y\x}+\sigma^{\x\z\y}-\sigma^{\y\z\x}, \nonumber\\
q_{[4]} &=&  -2 \Delta ^2 (\sigma^{\z\x\x\z} + \sigma^{\z\y\y\z})-(2 \Delta ^2+2) (\sigma^{\x\x00}+ \sigma^{\y\y00})+2 \Delta ( \sigma^{\x0\x0}+  
   \sigma^{\x\y\x\y}\nonumber\\&-&  \sigma^{\x\y\y\x}+  \sigma^{\x\z\x\z}+ \sigma^{\y0\y0}-\sigma^{\y\x\x\y}+ \sigma^{\y\x\y\x}+ \sigma^{\y\z\y\z}+ 
   \sigma^{\z0\z0}+ \sigma^{\z\x\z\x}+\sigma^{\z\y\z\y}\nonumber\\&-&2  \sigma^{\z\z00})-2 \sigma^{\x\z\z\x}-2 \sigma^{\y\z\z\y}, \nonumber\\   
q_{[5]} &=&  (4 \Delta ^3+14 \Delta )( 
\sigma^{\x\y\z00}-\sigma^{\y\x\z00}+\sigma^{\z\x\y00}-\sigma^{\z\y\x00}) 
+ 6 \Delta ^2 (-\sigma^{\x0\y\z0}+\sigma^{\x\y\x\x\z}\nonumber\\ &+& \sigma^{\x\y\y\y\z}+ \sigma^{\y0\x\z0}-\sigma^{\y\x\x\x\z}-\sigma^{\y\x\y\y\z}- \sigma^{\z\x0\y0}+\sigma^{\z\x\x\x\y}-\sigma^{\z\x\x\y\x}-\sigma^{\z\x\z\y\z}\nonumber\\&+&
   \sigma^{\z\y0\x0}+\sigma^{\z\y\y\x\y}-\sigma^{\z\y\y\y\x}+\sigma^{\z\y\z\x\z})+(10 \Delta ^2+8)(-\sigma^{\x\z\y00}+ \sigma^{\y\z\x00})\nonumber\\&+&
   6 \Delta ( \sigma^{\x0\z\y0}- \sigma^{\x\y0\z0}- \sigma^{\x\y\x\z\x}- \sigma^{\x\y\y\z\y}+ \sigma^{\x\z0\y0}- 
   \sigma^{\x\z\x\x\y}+ \sigma^{\x\z\x\y\x}+ \sigma^{\x\z\z\y\z}\nonumber\\&-& \sigma^{\y0\z\x0}+ \sigma^{\y\x0\z0}+ \sigma^{\y\x\x\z\x}+ \sigma^{\y\x\y\z\y}-  \sigma^{\y\z0\x0}- \sigma^{\y\z\y\x\y}+ \sigma^{\y\z\y\y\x}- \sigma^{\y\z\z\x\z}\nonumber\\
   &-& \sigma^{\z0\x\y0}+ \sigma^{\z0\y\x0}+ 
   \sigma^{\z\x\z\z\y}- \sigma^{\z\y\z\z\x})-6 \sigma^{\x\z\z\z\y}+6 \sigma^{\y\z\z\z\x},\nonumber\\
p_{[2]} &=& -2 q_{[3]},\nonumber\\
p_{[3]} &=&  2\Delta ^2 (\sigma^{\z\x\x\z}+\sigma^{\z\y\y\z})-(2\Delta ^2+2)(\sigma^{0\x\x0}+\sigma^{0\y\y0})+2 \Delta(-2\sigma^{0\z\z0}- 
   \sigma^{\x\y\x\y}
   \nonumber\\&+&  \sigma^{\x\y\y\x}- \sigma^{\x\z\x\z}+  \sigma^{\y\x\x\y}- \sigma^{\y\x\y\x}- \sigma^{\y\z\y\z}- \sigma^{\z\x\z\x}- \sigma^{\z\y\z\y}) + 2\sigma^{\x\z\z\x}+2\sigma^{\y\z\z\y}, \nonumber\\
p_{[4]} &=& (4 \Delta ^3+4 \Delta )( \sigma^{0\x\y\z0}-\sigma^{0\y\x\z0}-\sigma^{\z\x\y00}+\sigma^{\z\y\x00})+4 \Delta ^2(-\sigma^{\x\y\x\x\z}-\sigma^{\x\y\y\y\z} \nonumber\\&+&\sigma^{\y\x\x\x\z}+\sigma^{\y\x\y\y\z}+
   \sigma^{\z\x0\y0}-\sigma^{\z\x\x\x\y}+\sigma^{\z\x\x\y\x}+\sigma^{\z\x\z\y\z}-\sigma^{\z\y0\x0}-\sigma^{\z\y\y\x\y}\nonumber\\&+&
   \sigma^{\z\y\y\y\x}-\sigma^{\z\y\z\x\z})+(4 \Delta ^2+4)(-\sigma^{0\x\z\y0}+\sigma^{0\y\z\x0}+ 
   \sigma^{\x\z\y00}-\sigma^{\y\z\x00}) \nonumber\\&+&
   4 \Delta (  2\sigma^{0\z\x\y0}-2\sigma^{0\z\y\x0}+ \sigma^{\x\y0\z0}+ \sigma^{\x\y\x\z\x}+ \sigma^{\x\y\y\z\y}-2\sigma^{\x\y\z00}-\sigma^{\x\z0\y0}\nonumber\\&+&\sigma^{\x\z\x\x\y}-\sigma^{\x\z\x\y\x}-\sigma^{\x\z\z\y\z}-\sigma^{\y\x0\z0}-\sigma^{\y\x\x\z\x}-\sigma^{\y\x\y\z\y}+2\sigma^{\y\x\z00}+\sigma^{\y\z0\x0}\nonumber\\&+&\sigma^{\y\z\y\x\y}-\sigma^{\y\z\y\y\x}+
    \sigma^{\y\z\z\x\z}-\sigma^{\z\x\z\z\y}+\sigma^{\z\y\z\z\x})+4 \sigma^{\x\z\z\z\y}-4 \sigma^{\y\z\z\z\x},
\end{eqnarray}
where we use the notation $\sigma^{s_1 s_2\ldots s_k} \equiv \sigma^{s_1}\otimes \sigma^{s_2}\otimes \cdots \sigma^{s_k}$, $s_l \in\{0,\x,\y,\z,\pm\}$.
It has been shown in Ref. \cite{Zotos} that $Q_{[k]}$ can be applied to the Mazur inequality \cite{Mazur,Suzuki} with periodic boundary conditions in mind, to yield finite, non-vanishing spin Drude weights
for a general $XXZ$ model with a transverse magnetic field of strength $\chi$ added to the Hamiltonian density\footnote{
Note that considering non-zero transverse field $\chi$ is -- in TL -- equivalent to considering symmetry sectors with
non-vanishing total (conserved) magnetization. }, 
$h'=h + \chi \sigma^\z_0$,  at finite (non-zero) or even infinite temperature (where the evaluation of expectations $\omega_\beta(A)$ becomes easiest). This interesting result called for deeper theoretical understanding, as the Suzuki's proof \cite{Suzuki} only allows to consider the {\em incorrect} order of limits, namely $t\to\infty$ first, for a finite periodic system, and only then $n\to\infty$.
This problem has been now settled in the present paper.

Furthermore, it has been pointed out in \cite{Zotos}, that for a vanishing external field $\chi=0$, the Mazur bound on the spin Drude weight always vanishes, for any set of $Q_{[k]}$ from (\ref{eq:Qk}). This results from opposite `spin-flip' symmetry $\hat{\cal S} : \sigma^\z \to - \sigma^\z, 
\sigma^\pm \to \sigma^\mp$ of the spin-current, $\hat{\cal S}j = -j\hat{\cal S}$, and of the conserved charges, $\hat{\cal S}q_{[k]} = q_{[k]}\hat{\cal S}$. Non-vanishing magnetic field breaks the spin-flip symmetry of the Hamiltonian density and makes it possible that coefficients $w_k$ (\ref{eq:wk}) are non-vanishing and the bound (\ref{eq:generalbound}) is strictly positive.
On the other hand, there has been clear numerical evidence (see e.g. \cite{Meisner}) suggesting finite Drude weight and ballistic spin transport
at any temperature even at vanishing magnetic field strength $\chi=0$, in the easy-plane regime $|\Delta| < 1$.
This seemed to suggest that another, nontrivial quasi-local conservation law should exist with the same spin-flip symmetry as that of a spin current.

Indeed such a missing translationally invariant quasi-local conservation law $Q$ with negative spin-flip symmetry has recently been found~\cite{ProsenXXZ}, in case when $|\Delta| < 1$, satisfying almost commutation condition
\begin{equation}
[H_{\Lambda_n},Q_{\Lambda_n}] = 
-2\ii \sigma^\z_1 + 2\ii \sigma^\z_n.
\end{equation}
It has been shown that $Q$ admits a simple matrix-product representation in terms of infinite rank, almost-diagonal matrix operators
$\mm{A}_0,\mm{A}_{\pm}$, acting on an auxiliary Hilbert space with orthonormal basis labeled as
$\{ \ket{\La},\ket{\Ra},\ket{1},\ket{2},\ldots \}$,
\begin{eqnarray}
\mm{A}_0 &=& 
\ket{\La}\bra{\La} + \ket{\Ra}\bra{\Ra} + \sum_{r=1}^\infty \cos\left(r\varphi\right) \ket{r}\bra{r}, \nonumber \\
\mm{A}_+ &=& \ket{\La}\bra{1} + \sum_{r=1}^\infty \sin\left(2\left\lfloor \frac{r\!+\!1}{2}\right\rfloor \varphi\right) \ket{r}\bra{r\!+\!1},\label{eq:explicitA}\\
 \mm{A}_- &=& \ket{1}\bra{\Ra} - \sum_{r=1}^\infty \sin\left(\!\left(2\left\lfloor \frac{r}{2}\right\rfloor\!+\!1\right)\varphi\right)\ket{r\!+\!1}\bra{r}, \nonumber
\end{eqnarray}
where $\Delta = \cos \varphi$.
Namely, $Q$ satisfies the conditions of the Definition 1, with local densities of order $d\ge 2$, ($q^{(1)}=0$), generated as
\begin{eqnarray}
q^{(d)} &=& \ii\!\!\!\!\sum_{s_2,\ldots,s_{d-1}\in\{0,\pm\}}\!\!\!\bra{\La}\mm{A}_+\mm{A}_{s_2}\cdots \mm{A}_{s_{d-1}} \mm{A}_{-}\ket{\Ra} \label{eq:QZ}\\
&&\qquad\qquad\times\quad (\sigma^{+}\otimes\sigma^{s_2 \ldots s_{d-1}}\otimes\sigma^{-} - \sigma^{-}\otimes\sigma^{(-s_2) \ldots (-s_{d-1})}\otimes\sigma^{+}).
\nonumber
\end{eqnarray}
For example, the first few orders read explicitly:
\begin{eqnarray}
q^{(2)} &=& \ii(\sigma^{+-}-\sigma^{-+}) = {\textstyle\frac{1}{4}}j,\\
q^{(3)} &=& \ii \Delta ( \sigma^{+0-}-\sigma^{-0+}),\nonumber\\
q^{(4)} &=& \ii \Delta^2(\sigma^{+00-}-\sigma^{-00+}) + 2\ii \Delta (\Delta^2-1)(\sigma^{++--}-\sigma^{--++}). \nonumber
\end{eqnarray}
For the resonant values of the anisotropy $\Delta = \cos(\pi l/m)$, for coprime $l,m\in \ZZ$, $m> 1$, the rank of matrices (\ref{eq:explicitA})
becomes finite, i.e. it is $m+1$, whence the norm of $q^{(d)} $ can be estimated  $\| q^{(d)} \| < \gamma e^{-\xi d}$ by powering the
following rank-$(m+1)$ {\em transfer matrix}
\begin{eqnarray}
\mm{T} &=& \ket{\La}\bra{\La} +\ket{\Ra}\bra{\Ra} + \frac{1}{2}(\ket{\La}\bra{1}+\ket{1}\bra{\Ra})
+ \sum_{r=1}^{m-1} \biggl\{ \cos^2\!\left( \frac{\pi r l}{m}\right) \ket{r}\bra{r} \\ 
&+& \frac{1}{2}\sin^2\!\left(\!2\left\lfloor\!\frac{r\!+\!1}{2}\!\right\rfloor\!\frac{\pi l}{m}\right) \ket{r}\bra{r\!+\!1} + \frac{1}{2} \sin^2\!\left(\!\left(2\left\lfloor \frac{r}{2}\right\rfloor\!+\!1\right)\frac{\pi l}{m}\right)\ket{r\!+\!1}\bra{r}\biggr\},  
\nonumber
\end{eqnarray}
whose subleading eigenvalue is strictly smaller than $1$ \cite{ProsenXXZ}.
In this regime, one can use (\ref{eq:QZ}) in our Theorem 1 to bound $D_\beta$ even for a vanishing transverse magnetic field $\chi=0$.
The calculation can be made exact for infinite temperature $\beta\to 0$, where we present a rigorous lower bound
expressed as \cite{ProsenXXZ}
\begin{equation}
\lim_{\beta\to 0} \frac{D_\beta}{\beta} \ge 4 D_Z, \quad D_Z := \frac{1}{4} \lim_{n\to \infty} \frac{n}{\bra{\La}\mm{T}^n\ket{\Ra}}.
\label{eq:exactbound}
\end{equation}
This expression can be evaluated explicitly \cite{Affleck} in terms of Jordan decomposition of $\mm{T}$, yielding a {\em fractal} (nowhere continuous) dependence on $\Delta$:
\begin{equation}
D_Z = \frac{1}{2}(1-\Delta^2) \frac{m}{m-1}, \qquad \Delta = \cos\frac{\pi l}{m}.
\end{equation}

Note that two marginal cases with $|\Delta|=1$, (for $l=0$ and $l=m$), are exceptional, as translationally invariant operator
\begin{equation}
Q(\Delta=\pm 1)=\ii\sum_{d=2}^{\infty}\Delta^{d-2}\sum_{x\in \mathbb{Z}}\eta_x(\sigma_1^+ \sigma_d^- - \sigma_1^- \sigma_d^+)
\end{equation}
is no longer a spatial sum of exponentially localized operators (therefore incompatible with the Definition~\ref{def:1}), i.e. $\| q^{(d)}\|$ ceases to decay with increasing $d$.

Computing the Drude weight bounds for a finite temperature ($\beta > 0$) is certainly more tedious and cannot be done as explicitly as for the infinite temperature (\ref{eq:exactbound}). One possible approach for high-temperature is, for example, a $\beta$-expansion \cite{Enej}.

We stress that the bounds provided by our Theorems 1 and 2 for the $XXZ$ model are mathematically rigorous if one assumes in addition the existence of the
Drude weight. The second assumption in the Theorem 1 (ii), essentially requiring an (arbitrarily slow) relaxation of the extended current-current temporal correlations 
$c(t)-\bar{c}$, can be understood as necessary for having a uniquely defined Drude weight (see the discussion in subsect.~\ref{sect:lr}).

\section{Discussion}

\label{sect:discussion}

\subsection{On the linear response derivation of the Drude weight}
\label{sect:lr}

Note that a slightly different form of the Drude weight than (\ref{eq:Dbeta}) follows from a strict derivation of the linear response where the canonical Kubo-Mori inner product 
$\int_0^\beta \dd \lambda \omega_\beta(a^* \tau_{\ii \lambda}(b))$ replaces a simple thermal average $\omega_\beta(a^*b)$.  We will show that a small extra assumption of  ``non-ergodic dynamical mixing'' is needed in order to justify the simple thermal-averaged expression (\ref{eq:Dbeta}). Let us here carefully outline our linear response setup and the main steps of our argument.

We consider the constant-gradient field perturbation to the Hamiltonian, which extends on a finite symmetric sublattice $[-n,n]$ (for $\Lambda  \supseteq [-n,n]$), as
\begin{equation}
H^{F,n}_{\Lambda} = H_\Lambda - F \sum_{x=-n}^{n-d_q+1} \! x\, q_x,
\quad {\rm for\; some}\quad q \in \AAA_{[0,d_q-1]}.
\label{eq:quenched}
\end{equation}
$H^{F,n}_{\Lambda\to\ZZ}$ generates  a perturbed time evolution $\tau^{F,n}_t(a) = \lim_{\Lambda\to\ZZ} e^{\ii H^{F,n}_\Lambda} a  e^{-\ii H^{F,n}_\Lambda}$, on an infinite lattice, but still for a finite {\em field extension} $[-n,n]$. 
Starting in the equilibrium state $\omega_\beta$ for $F=0$, at $t=0$, and then {\em quenching} the Hamiltonian by switching on the force field (\ref{eq:quenched}), we define the {\em canonical} Drude weight as the {\em asymptotic rate} at which the local current in the bulk increases per unit time, after taking the limit of infinite field extension:
\begin{equation}
\tilde{D}_\beta:= \lim_{t\to \infty}\frac{1}{2t}\lim_{n\to \infty}\left[\frac{\dd}{\dd F}\omega_{\beta}\left(\tau^{F,n}_{t}(j)\right)\right]_{F=0}.
\end{equation}
Note that, in contrast to Ref.~\cite{Jaksic}, we have to take a vanishing force limit $F\to 0$ first, in order to make our perturbation bounded (and well defined) in the infinite extension limit $n\to\infty$.
Writing in the first order of Born/Dyson expansion
\begin{equation}
(\tau_{-t}\circ \tau^{F,n}_t)(j) = j - \ii F \int_0^t \dd s \sum_{x=-n}^{n-d_q+1}\!x [\tau_{-s}(q_x), j] + {\cal O}(F^2),
\end{equation}
and observing the following identity, for any $f,g\in\AAA$
\begin{equation}
\omega_\beta([\tau_t(f),g]) = \ii \int_0^\beta\!\dd \lambda\,\omega_\beta(\tau_{t-\ii \lambda}(\delta(f)) g)
\end{equation}
where $\delta(f)\equiv \lim_{\Lambda\to\ZZ} \ii [H_\Lambda,f]$ is  the $*$-derivation with respect to the
the unperturbed dynamics,
we arrive at
\begin{equation}
\tilde{D}_\beta = \lim_{t\to\infty}\frac{1}{2t}\lim_{n\to\infty}\int_{0}^{t}\dd s\int_{0}^{\beta}\dd\lambda \sum_{x=-n}^{n-d_q+1} x\, \omega_{\beta}
\left(\tau_{-s-\ii\lambda}\left(\delta(q_x)\right)j\right).
\end{equation}
Finally, we choose the current $j$ and the force-field density $q$, such that they satisfy the {\em continuity equation} $j_{x-1}-j_x=\delta(q_x)$ (see, as an example, Eq.~(\ref{eq:rec1})). The sum over $x$ can now be expressed as a difference of two sums 
which result in, after shifting the index $x\to x + 1$ in the first sum containing $j_{x-1}$,
\begin{equation}
\tilde{D}_\beta = \lim_{t\to\infty}\frac{1}{2t}\lim_{n\to\infty}\int_{0}^{t}\dd s\int_{0}^{\beta}\dd\lambda \sum_{x=-n}^{n-d_j+1} \omega_{\beta}
\left(\tau_{-s-\ii\lambda}\left(j_x\right)j\right),
\label{eq:Dtilde1}
\end{equation}
except for the boundary terms  of magnitude $|\omega_{\beta}(\tau_{-s-i\lambda}(j_x)j)|$ at $x\approx \pm n$ which decay exponentially in $n$ and uniformly in $z = -s-\ii \lambda$ (and thus do not contribute in the limit $n\to\infty$ of (\ref{eq:Dtilde1}))
according to the Theorem 4.2 of Araki \cite{Araki}: Namely, 
$\lim_{n\to \infty} e^{|n|\rho}\| [f,\tau_z(\eta_n (g))]\|=0$, applying for any strictly local $f,g$ and $z\in\CC$.
Shifting the time evolution to the second factor of RHS (\ref{eq:Dtilde1}), 
the sequence $\tilde{c}_n(z):=\sum_{x=-n}^{n-d_j+1}\omega_{\beta}\left(j_x \tau_{z}(j)\right)$ converges uniformly to $c(z)$ (expression (\ref{eq:cdef1}) of Theorem 1, but for complex time argument $z$), again as a consequence of the above
mentioned Araki's proof of locality of complex time evolution. Thus, the limit $n\to\infty$ and $s,\lambda$ integrations can be interchanged. Finally, we obtain
\begin{equation}
\tilde{D}_\beta = 
\lim_{t\to\infty}\frac{1}{2t} \int_{0}^{t}\dd s\int_{0}^{\beta}\dd\lambda\, c(s + \ii \lambda).
\label{eq:Dtilde2}
\end{equation}
Clearly, a sufficient condition for the equality between the {\em thermal} (\ref{eq:Dc}) and the {\em canonical} (\ref{eq:Dtilde2}) 
Drude weights
\begin{equation}
D_\beta = \tilde{D}_\beta
\end{equation}
is a kind of {\em non-ergodic weak-mixing} condition
\begin{equation}
\lim_{t\to \infty} \frac{1}{2t} \int_{0}^t \dd s |c(s + \ii \lambda) - \bar{c}| = 0, \quad \forall \lambda \in [0,\beta].
\label{eq:weakmixing}
\end{equation}
This follows from writing $|D_\beta(t) - \tilde{D}_\beta(t)| \le \frac{1}{2t}\int_{0}^t \dd s \int_0^\beta \dd\lambda |c(s + \ii\lambda) - \bar{c}|$ where $D_\beta(t)$ and $\tilde{D}_\beta(t)$ denote the expressions (\ref{eq:Dbeta}) and (\ref{eq:Dtilde2}), respectively, before taking the limit $t\to\infty$.
For example, for the condition (\ref{eq:weakmixing}) to be fulfilled it is enough that the correlation function $c(z)$ converges 
(to an asymptotic value $\bar{c}$, no matter how slowly), as $|z|\to\infty$ anywhere within the thermal strip ${\rm Im\,}z\in [0,\beta]$.
This is likely to be expected as a generic property of many-body systems with thermodynamically dense spectra, however it seems one may not trivialize it using Cauchy's integral theorem, as in a related study of Ref.~\cite{Jaksic} with holomorphic dynamics of finite-system's observables. For instance, our $c(z)$ may not be holomorphic (due to an infinite extension limit $n\to \infty$ needed for a definition (\ref{eq:cdef1})). In fact, a better understanding of analytic properties of $c(z)$ seems to be an interesting problem for future investigations.

Nevertheless, even if $D_\beta$ (\ref{eq:Dbeta}) may not be both, rigorously and generally connected to the linear response theory of the Drude weight, it might still be interesting general-purpose quantity in its own right, as a simple dynamical indicator of non-ergodic $C^*$ dynamical systems possessing non-trivial conservation laws.

\subsection{General remarks}

Note that our Theorems~\ref{theo:main},\ref{theo:general} can be applied to any completely integrable quantum chain, like the one-dimensional fermionic Hubbard model, supersymmetric $t$-$J$ model, etc., for which one has algebraic procedures (e.g. algebraic Bethe Ansatz) for obtaining nontrivial conserved quantities which can be written as exponentially convergent sums of local operators.
It is enough that these operators are conserved only in the TL, while for finite chains their time-derivative is supported only on finite domains near the boundaries of the chain. 

We conclude by remarking that the conditions of our theorems could be relaxed or generalized in some cases: (i) The translational invariance of interactions was required only to guarantee a sufficiently strong clustering in the equilibrium state. This can perhaps be relaxed, and replaced by an algebraic clustering in some instances (with inverse power larger than $1$, e.g. for zero temperature equilibrium states in the so-called {\em critical} systems), or exponential clustering may set in for other reasons (e.g. due to {\em disordered} interactions and {\em localization}, however in such cases we are not aware of non-trivial examples of extensive conserved quantities). (ii) Quasi-local almost-conserved quantities need not be exponentially localized, but it would be enough to have a power law bound $\| q^{(d)}\| \le \gamma d^{-\nu}$ for sufficiently large inverse power $\nu > 3$. (iii) Similar theory could be built for  higher, $D$-dimensional quantum spin lattices, where time-derivatives (commutators) of almost conserved quantities could result in terms distributed on the $(D-1)$-dimensional boundary of the lattice. Again, we do not see yet any application of this potentially very interesting generalization.
\\\\
The work has been partially supported by research grants P1-0044 and J1-2208 of Slovenian Research Agency (ARRS).

\end{document}